\documentclass[letterpaper,10 pt,conference]{ieeeconf} 
\IEEEoverridecommandlockouts 
\usepackage{graphicx}
\usepackage{euscript,color}
\usepackage{amssymb,amsfonts,amsmath,amscd,dsfont,mathrsfs}
\usepackage{cite}

\DeclareMathOperator*{\argmin}{arg\,min}

\newcommand{\Ac}{\mathcal{A}}
\newcommand{\Bc}{\mathcal{B}}

\newcommand{\Ec}{\mathcal{E}}

\newcommand{\Nc}{\mathcal{N}}

\newcommand{\Sc}{\mathcal{S}}

\newcommand{\Uc}{\mathcal{U}}

\newcommand{\xv}{{\bf x}}
\newcommand{\yv}{{\bf y}}

\newcommand{\xh}{{\hat{x}}}

\newcommand{\eh}{{\hat{e}}}

\newcommand{\et}{{\tilde{e}}}

\newcommand{\xt}{{\tilde{x}}}

\def\d{\delta}
\def\e{\epsilon}

\DeclareMathOperator\E{E}
\let\P\relax
\DeclareMathOperator\P{P}
\newcommand{\ind}{\mathbf{I}}

\newcommand\ie{i.e.,\;}
\newcommand{\Bern}{\mathrm{Bern}}

\newtheorem{definition}{Definition}

\newtheorem{theorem}{Theorem}
\newtheorem{lemma}{Lemma}

\newtheorem{corollary}{Corollary}

\newtheorem{proposition}{Proposition}

\begin{document}

\title{\LARGE \bf Minimum Complexity Pursuit}

\author{Shirin Jalali and Arian Maleki
\thanks{S. Jalali is a postdoctoral scholar at the Center for Mathematics of Information, California Institute of Technology, Pasadena, CA,
        {\tt\small shirin@caltech.edu}}%
\thanks{A. Maleki is a postdoctoral scholar at Digital Signal Processing group, Rice University, Houston, TX,
        {\tt\small arian.maleki@rice.edu}}%
}

\maketitle

\newcommand{\p}{\mathds{P}}
\newcommand{\mb}{\mathbf{m}}   
\newcommand{\bb}{\mathbf{b}}

\begin{abstract}
The fast growing field of compressed sensing is founded on the fact that if a signal is ÔsimpleÕ and has some `structure', then it can be reconstructed accurately with far fewer samples than its ambient dimension. Many different plausible structures have been explored in this field, ranging from sparsity to low-rankness and to finite rate of innovation. However, there are important abstract questions that are yet to be answered. For instance, what are the general abstract meanings of ÔstructureÕ and ÔsimplicityÕ? Does there exist universal algorithms for recovering such simple structured objects from fewer samples than their ambient dimension? In this paper, we aim to address these two questions. Using algorithmic information theory tools such as Kolmogorov complexity, we provide a unified method of describing ÔsimplicityÕ and ÔstructureÕ. We then explore the performance of an algorithm motivated by OcamÕs Razor (called MCP for minimum complexity pursuit) and show that it requires $O(k\log n)$ number of samples to recover a signal, where $k$ and $n$ represent its complexity and ambient dimension, respectively. Finally, we discuss more general classes of signals and provide guarantees on the performance of MCP.
\end{abstract}


\section{Introduction}\label{sec:intro}
Compressed sensing (CS) refers to a body of techniques that undersample high-dimensional signals,  and yet recover them accurately by exploiting their intrinsic `structure' \cite{Donoho1, CaRoTa06}. This permits more efficient sensing systems that are proved to be valuable in many applications including magnetic 
resonance imaging (MRI) \cite{LuDoSaPa08} and radar \cite{HeSt09}, to name a few. Some of the `structures' that have been considered in the literature are as follows. 
\begin{itemize}
\item[i.]  Sparsity:  A vector $x \in \mathds{R}^n$
is called $k$-sparse if and only if $\|x\|_0 \triangleq \sum_{i=1}^n \ind_{\{x_i \neq 0\}} \leq k$. Roughly speaking, 
according to compressed sensing a $k$-sparse signal $x$ can be recovered from $d=O(k\log n)$ random linear
measurements $y = Ax$.

\item [ii.] Low rankness: If $X \in \mathds{R}^{m \times n}$ is a low rank matrix with
${\rm rank}(X) \leq k$, then $d = O(r(m+n) \log(mn))$ random linear measurements are sufficient for recovering $X$ from its measurements 
accurately with high probability \cite{ReFaPa10}.
\item[iii.] Model-based compressed sensing: \cite{RichModelbasedCS} considers more structured signal models by assuming 
that from ${n \choose k}$ subspaces of $k$-sparse signals only $m_k$ of them may occur. It is then proved that $O(\log(m_K))$ random linear measurements are sufficient for the accurate recovery of such signals. This class is a superset of some of the other structures introduced in the literature such as the class of block-sparse signals  \cite{eldar2010block, stojnic2009reconstruction, stojnic2009block, MaCeWi05}.

\item[iv.] Rate of innovation: \cite{VeMaBl02} defines the rate of innovation of a signal as its ``degrees of freedom''. Several important classes of functions such as the piecewise polynomial functions and sparse signals have clearly finite rate innovation. \cite{VeMaBl02}  suggests sampling schemes for several classes that recover the signal from $O(k)$ number of measurements, where $k$ is the rate of innovation.
 \end{itemize}
 
The above results seem to provide pieces of a bigger picture. Recently,  \cite{ChRePaWi10} introduced the  class of simple functions and atomic norm as a framework that unifies some of the above observations and extends them  to some other signal classes. However, there is still an interesting conceptual question that needs to be addressed, \ie what is the 
abstract meaning of `structure' that allows fewer measurements than the ambient dimension of the signal? Given a simple signal, 
which scheme recovers the signal from an undersampled random linear set of measurements?\\

In the context of algorithmic information theory, Solomonoff \cite{Solomonoff} and Kolmogorov \cite{KolmogorovC} 
suggested a universal notion of complexity for binary sequences, known as the Kolmogorov complexity. 
Given a binary sequence $x$, its Kolmogorov complexity $K(x)$ is defined as the length of the shortest computer program that prints $x$. 
In this paper, we extend the concept of Kolmogorov complexity to the real signals. Such extensions are straightforward and have been explored
before \cite{Staiger2002455}. Based on this notion of complexity, called
Kolmogorov complexity of real signals, we show that OccamÕs razor\cite{ocam}, i.e., finding the `simplest' solution of the linear equations, 
correctly recovers the signal with much fewer measurements than the ambient dimension of the signal. 
Roughly speaking, we prove that the number of linear measurements required for recovering the correct solution is proportional to the complexity
rather than the ambient dimension of the signal. We postpone the accurate exposition of our results to Section \ref{sec:contrib}. 
We will further discuss the issue of model mismatch in the signal classes and will prove
that the approach motivated by OccamÕs razor is stable with respect to such non-idealities in the system.\\

Here is the organization of our paper. Section \ref{sec:def} defines  the notation used throughout the paper. Senction \ref{sec:kolm} defines Kolmogorov complexity of a real-valued signal. Section \ref{sec:contrib} outlines our contribution. Section \ref{sec:examp} calculates the Kolmogorov complexity of several
classes that are popular in compressed sensing and clarifies the statements of our theorems on these classes. Section \ref{sec:related}
compares our work with other results in the literature. Sections \ref{sec:proof1} and \ref{sec:proof2} are devoted to the proofs of
our main theorems.


\section{Definitions}\label{sec:def}
Calligraphic letters such as $\Ac$ and $\Bc$ denote sets. For a set $\Ac$, $|\Ac|$ and $\Ac^c$ denote its size and its complement, respectively. 
For a sample space $\Omega$ and  event set $\Ac\subseteq \Omega$, $\ind_{\Ac}$ denotes the indicator function of the event $\Ac$.

Let $\{0,1\}^*$ denote the set of all finite-length binary sequences, \ie $\{0,1\}^*\triangleq\cup_{n\geq 1}\{0,1\}^n$. For a vector $x\in\mathds{R}^n$, the $\ell_p$ norm of $x$ is defined as $\|x\|_p\triangleq(\sum_{i=1}^n|x_i|^p)^{1/p}$. The $\ell_{\infty}$ norm of $x$ is denoted by $\|x\|_{\infty}\triangleq\max_{i}|x_i|$.

For a real number $x\in[0,1]$, let $[x]_m$ denote the $m$-bit approximation of $x$ that results  from taking the first $m$ bits in the binary expansion of $x$. In other words, if   $x=\sum_{i=1}^{\infty}2^{-i}(x)_i$, where $(x)_i\in\{0,1\}$ denotes the $i^{\rm th}$ bit in the binary expansion of $x$,  then
\begin{align}
[x]_m\triangleq\sum_{i=1}^{m}2^{-i}x_i.
\end{align}
Similarly, for a vector $x^n\in[0,1]^n$, define
\begin{align}
[x^n]_m\triangleq ([x_1]_m,\ldots,[x_n]_m).
\end{align}
For an integer $n\in\mathds{N}$, let  
\[
\log^*n \triangleq \lceil \log_2 n\rceil + 2\log_2 \max(\lceil \log_2 n\rceil ,1).
\]


\section{Kolmogorov complexity}\label{sec:kolm}
The Kolmogorov complexity of a finite-length sequence $x$ with respect to a \emph{universal computer} $\Uc$ is defined as the minimum length over all programs that print $x$ and halt.\footnote{Refer to Chapter 14 of \cite{cover} for the exact definition of a universal computer, and more details on the definition of the Kolmogorov complexity.}  For a universal computer $\Uc$ and any computer $\Ac$, there exists a constant $c_{\Ac}$ such that $K_{\Uc}(x)\leq K_{\Ac}(x)+c_{\Ac}$, for all strings $x\in\{0,1\}^*$ \cite{cover}. Hence, as suggested in \cite{cover}, we drop the subscript $\Uc$, and let $K(x)$ denote the Kolmogorov complexity of the binary string $x$.

Similarly, the Kolmogorov complexity of an integer $n\in\mathds{N}$, $K(n)$, is defined as the Kolmogorov complexity of its binary representation. It can be proved that 
\[
K(n)\leq \log^* n+c,
\]
where $c$ is a constant independent of $n$.

\noindent For $x=(x_1,x_2,\ldots,x_n)\in [0,1]^n$, define the Kolmogorov complexity of $x$ at  resolution $m$ as
\begin{align}
K^{[\cdot]_m}(x)  = K([x_1]_m,[x_2]_m,\ldots,[x_n]_m).
\end{align}

\begin{lemma} For $(x_1, x_2, \ldots, x_n) \in [0,1]^n$,
\[
\limsup_{m \rightarrow \infty} \frac{K^{[\cdot]_m}(x_1,x_2, \ldots, x_n)}{m} \leq n.
\]
\end{lemma}
The proof is very simple and is skipped.
\begin{definition}
The signal $x=(x_1, x_2, \ldots, x_n)$ is called incompressible if and only if 
\[
\lim_{m \rightarrow \infty} \frac{K^{[\cdot]_m}(x_1,x_2, \ldots, x_n)}{m} = n.
\]
\end{definition}
\begin{proposition}
Let $\{X_i\}_{i=1}^n \overset{iid}{\sim} U[0,1]$. Then,
\[
\frac{1}{m}K^{[\cdot]_m}(X_1, X_2, \ldots X_n) \rightarrow n
\]
in probability.
\end{proposition}
\begin{proof}
If $X_i=\sum_{j=1}^{\infty} (X_i)_j2^{-j}$, where $(X_i)_j\in\{0,1\}$, then $\{(X_i)_j\}_{j=1}^{\infty} \overset{iid}{\sim} \Bern(1/2)$. Theorem 14.5.3 in \cite{cover} states that  the normalized Kolmogorov's complexity of $([X_1]_m,\ldots,[X_n]_m)={\{((X_i)_1,(X_i)_2,\ldots,(X_i)_m)\}_{i=1}^{n}}$, \ie
\begin{align}
{K(\{(X_i)_1,(X_i)_2,\ldots,(X_i)_m\}_{i=1}^n|mn) \over mn}\to 1,\label{eq:K1}
\end{align}  
in probability.  On the other hand,
\begin{align}
&K(\{(X_i)_1,(X_i)_2,\ldots,(X_i)_m\}_{i=1}^n|mn)\nonumber\\
&\leq K(\{(X_i)_1,(X_i)_2,\ldots,(X_i)_m\}_{i=1}^n)\nonumber\\
& \leq  K(\{(X_i)_1,(X_i)_2,\ldots,(X_i)_m\}_{i=1}^n|mn) +\log^*(mn) + c,\label{eq:K2}
\end{align}
where $c$ is a constant \cite{cover}. Hence, combing \eqref{eq:K1} and \eqref{eq:K2} proves the desired result.
\end{proof}


\section{Our contribution}\label{sec:contrib}
Consider the problem of reconstructing a vector $x_o \in \mathds{R}^n$ from $d$ random linear measurements $y = Ax$ with $d <n$. We say a recovery algorithm
is successful if as $n$ grows the $\ell_2$-error between $x_o$ and its reconstruction $\xh_o$ goes to zero, \ie we want
\[
\P\left (\|x_{o}^n-\xh_{o}^n\|_2^2 >\e \right)\to 0,
\]
for any $\e>0$. Assuming that the signal is `structured' in the
sense that will be clarified later, we follow Ocam's Razor and seek the simplest solution of $y = Ax$, i.e., 
\begin{eqnarray}
&&\arg \min\quad K^{[\cdot]_m}(x_1,\ldots,x_n)\nonumber \\
&&{\rm s.t.}\quad \ \  \;\;\;\; Ax^n = y_o^n.\label{eq:alg}
\end{eqnarray}
We call this algorithm minimum complexity pursuit or MCP. The choice of $m$ will be clarified later as well.  
Suppose that  $A\in\mathds{R}^{d\times n}$, where $A_{ij}$ are iid $\Nc(0,1/d)$, and assume that $y_o^n=Ax_{o}^n$. 
Let $\xh_o^n=\xh_o^n(y_o^n,A)$ denote the output of \eqref{eq:alg} to the inputs $y_o^n$ and $A$.

\begin{theorem}\label{thm:1}
Assume that $x_o=(x_{o,1},x_{o,2},\ldots)\in[0,1]^{\infty}$ is  such that
\begin{align}
\limsup_{n\to\infty} {K^{[\cdot]_{m}}(x_{o,1},x_{o,2},\ldots,x_{o,n}) \over m} \leq \kappa,\label{eq:limited_kol}
\end{align} 
where $m=m_n=\lceil \log n\rceil$. Let $d=d_n=\lceil\kappa \log n \rceil$.
Then, for any $\e>0$
\begin{align}
\P\left(\|x_{o}^n-\xh_{o}^n\|_2^2>\e\right)\to 0,
\end{align}
as $n$ grows without bound.

\end{theorem}

This theorem indicates that when the Kolmogorov complexity of the signal is less than $\kappa$, then
 $O(\kappa \log n)$ linear measurements are sufficient for the successful recovery. Also, it provides an
 evidence for the success of Ocam's Razor.

Although Theorem \ref{thm:1} is an asymptotic theorem, its proof provides information on the 
performance of MCP on finite length sequences as 
well.

\begin{corollary}
Assume that $x_o=(x_{o,1},x_{o,2},\ldots, x_{o,n})\in[0,1]^{n}$ is  such that
\begin{align*}
{K^{[\cdot]_{m}}(x_{o,1},x_{o,2},\ldots,x_{o,n}) \over m} \leq \kappa,  \ \ \forall \; m.  
\end{align*} 
Let $m=m_n= \lceil \alpha \log n\rceil$ and $d=d_n=\lceil 2 \alpha \kappa \log n \rceil$. Then, with probability $1-{n^{-\alpha \kappa}}$
\begin{align*}
\|x_{o}^n-\xh_{o}^n\|_2 \leq \frac{10 n^{1/2- \alpha}}{\sqrt{\kappa} \log n}.
\end{align*}

\end{corollary}


Now consider the following more general setting, where the original signal $x_o^n$ to be recovered is not low-complexity, but is close to a low-complexity signal $\xt^n$, i.e., $\|x_o^n- \xt^n\|_2 \leq \epsilon_n$ with $\epsilon_n = o(1)$. Again, let $y_o^n=Ax_o^n$, and consider the following reconstruction algorithm for finding $x_o^n$ from its linear measurements $y_o^n$:
\begin{eqnarray*}
&&\min\quad K^{[\cdot]_m}(x_1,\ldots,x_n)\\
&&{\rm s.t.}\quad \ \  \|Ax^n -y_o^n\|_2 \leq  \sigma_{max}(A) \epsilon_n.\label{eq:alg_model_mismatch}
\end{eqnarray*}

\noindent Assume that $A\in\mathds{R}^{d\times n}$ and $A_{ij}$ are iid~$\Nc(0,{1\over d})$. Let $\xh^n_o=\xh_o^n(y_o^n,A)$.

\begin{theorem}\label{thm:2}
Assume that there exists $\xt_o^n$ such that $\|x_o^n-\xt^n_o\|_2 \leq \epsilon_n$, and
\begin{align} 
\limsup_{m\to\infty}{K^{[\cdot]_m}(\xt^n_o)  \over m}\leq \kappa_n.
\end{align}
Let $m=m_n=\lceil \log n\rceil$ and $d=d_n=\lceil\kappa_n \log n \rceil$. If $\e_n = o(d_n/n)$,
then for each $\e>0$,
\begin{align}
\P\left(\|x_{o}^n-\xh_{o}^n\|_2^2>\e\right)\to 0,
\end{align}
as $n$ grows without bound.
\end{theorem}

In the next section we show that several popular classes of sequences studied in CS such as class of
sparse signals and samples of piecewise smooth functions can be considered as special cases
of the framework we introduced in this section and that Theorems \ref{thm:1} and \ref{thm:2} provide
useful information about them.


\section{Applications}\label{sec:examp}
 It is well-known that the Kolmogorov complexity is not computable. In fact, the only way to find the shortest program
 that generates a sequence is to run all the short programs and see if they generate the sequence or not. However, some
 short programs may not halt and there is no way to figure out if the program will halt or not. Hence, there is no effective
 way to calculate the Kolmogorov complexity. However, it is usually possible to find upper bounds for the Kolmogorov 
 complexity. In this section, we consider several popular examples and provide upper bounds for
 their Kolmogorov complexity. Based on these upper bounds we use Theorems \ref{thm:1} and \ref{thm:2}
 to calculate the number of random linear measurements required by the MCP to recover these functions.
 This demonstrates the connection between the results of Section \ref{sec:contrib} and the
 compressed sensing and finite rate of innovation frameworks explained in Section \ref{sec:intro}.  
 It is straightforward to extend the results to the other classes we discussed in Section \ref{sec:intro}.

\subsection{Sparsity}  

Let the signal $x_o=(x_{o,1}, x_{o,2}, \ldots,x_{o,n})$ be $k$-sparse. Consider the following program for describing $[x_o^n]_m$. 
First, use a program of constant length to describe the structure of the signal as `sparse'  and the ordering of the rest of information. Then, spend $\log^* n+c$ bits to describe the
length of the signal. Next, code the sparsity level $k$ with $\log^* k$ bits, and spend $k (\log^*n+c)$ more bits to code the locations of the $k$ non-zero elements. Finally, use
$km$ more bits to describe the quantized magnitudes of the non-zero coefficients. Therefore, we have
\begin{align}\label{eq:kc_poly}
&{K^{[\cdot]_{m}}(x_{o,1},x_{o,2},\ldots,x_{o,n}) \over m}\nonumber \\
 &\leq k+ \frac{(k+1)(\log^*n + c)+\log^*k+c}{m}.
\end{align}
Plugging \eqref{eq:kc_poly} into Theorem \ref{thm:1}, we conclude that $\lceil (2k+1) \log n\rceil$ measurements are sufficient for the recovery of the $k$-sparse signals. 

\subsection{Piecewise polynomial} 

Let  $(x_{o,1}, x_{o,2}, \ldots,x_{o,n})$ be samples of a piecewise polynomial function $f(x)$ defined on $[0,1]$ at locations $(0, 1/n, \ldots,(n-1)/n)$. Further, assume that $0 \leq f(x)\leq 1$, for every $x$. Let ${\rm Poly}_N^Q$ represent the class of such functions which have at most $Q$  singularities\footnote{A singularity is a point at which the  function is not infinitely differentiable.} and $N$ is the maximum degree of each polynomial. Let $\{a_i^{\ell}\}_{i=0}^{N_{\ell}}$ denote the set of coefficients of the $\ell^{\rm th}$ polynomial, where $N_{\ell}\leq N $ denotes its degree. For the notational simplicity, we assume that the coefficients of each polynomial belong to the $[0,1]$ interval and that $\sum_{i=0}^{N_{\ell}} a^{\ell}_i <1$ for every $\ell$, where $a_i^{\ell}$ is the $i^{\rm th}$ coefficient of the $\ell^{\rm th}$ polynomial. For a given length $n$, we derive an upper bound on the Kolmogorov complexity.  Consider the following program for describing $[x_o^n]_m$. The code  first specifies the model as `piecewise polynomial' with parameters $(n,Q,N)$. This requires $\log^*n+ \log^*N + \log^* k + c_1$ bits. Then, for each singularity point, the code first determines the largest sampling point $i/n$ that is smaller than it. Since there are at most $Q$ singularity points, describing this information requires at most $Q(\log^* n+c_2)$ bits. The next step is to describe the coefficients of each polynomial. Using an $m'$-bit quantizer for each coefficient, the induced error is bounded by
              \begin{align}
              \left|\sum_{i=0}^{N_{\ell}} a^{\ell}_i t^n-\sum_{i=0}^{N_{\ell}} [a^{\ell}_i ]_{m'} t^n \right| &\leq \sum_{i=0}^{N_{\ell}} |a^{\ell}_i- [a^{\ell}_i]_{m'}|\nonumber\\
              & \leq (N+1) 2^{-m'}. 
              \end{align}
To ensure that we are able to reconstruct the $m$-bit resolution of the samples from this description, $(N+1) 2^{-m'}<2^{-m}$.  Therefore, describing the polynomials' coefficients we need  $(Q+1)(N+1) (m +\lceil\log_2(N+1)\rceil)$ extra bits. Hence, overall, we conclude that
               \begin{align} \label{eq:kc_pp}
               &{K^{[\cdot]_{m}}(x_{o,1},x_{o,2},\ldots,x_{o,n}) \over m} \leq
                 (Q+1)(N+ 1)\nonumber\\
                 &+\frac{(Q+1)(N+ 1)\lceil\log_2(N+1)\rceil}{m}\nonumber \\
                 &+\frac{\log^*n+ \log^*N + \log^* k + Q\log^*n+ c_1+c_2}{m}.
               \end{align}
               It is straightforward to plug \eqref{eq:kc_pp} into Theorem \ref{thm:2} and prove that, roughly speaking, for large values of $n$, $(QN+ 2Q+1) \log n$ measurements are sufficient for the successful recovery of the piecewise polynomial functions.

So far we have considered examples of low-complexity signals. However, in many applications the signals are not of low complexity
but are rather close to low complexity signals. We present several examples here.

\subsection{$\ell_p$-constrained signals} While sparse signals have played an important role in the theory of compressed sensing, it is well-known 
that they do not occur in practice very often. More accurate models assume that either the magnitude of the signal follows a
specific decay or the signal belongs to an $\ell_p$ ball with $p<1$, i.e., $\|x_o\|_p \leq 1$ \cite{Donoho1, MalekiThesis}.
For the signal $x_o\in\mathds{R}^n$ with $\|x_o\|_p\leq 1$, let $(x_{o,(1)},x_{(2)}, \ldots,x_{o,(n)})$ denote the permuted version of $x_o$ such that $x_{o,(1)}\geq x_{o,(2)}\geq \ldots \geq x_{o,(n)}$. It is easy to show that $x_{o,(i)} \leq i^{-\frac{1}{p}}$. Therefore, if we just keep the
$k$ largest coefficients of this signal and set the rest to zero the resulting $k$-sparse vector $\tilde{x}_o$ satisfies, $\|x_o- \tilde{x}_o\| \leq  k^{-\frac{1}{p}+\frac{1}{2}} $. Setting the sparsity $k$ to $n^{p/2}$, Theorem \ref{thm:2} proves that $d_n = n^{p/2}\log n$ samples are sufficient for asymptotically accurate recovery. It is interesting to note that
as $p$ decreases, the decay rate increases and the number of measurements required for the successful recovery decreases. 
 
\subsection{Smooth functions} Suppose that $x_1, x_2, \ldots, x_n$ are equispaced samples of  a smooth function $f:[0,1] \rightarrow \mathds{R}$ with $0 \leq f(x) \leq 1$. 
Let the function be $\beta+1$ times differentiable and $\|f^{(\beta+1)}\|_{\infty} \leq \gamma$. For the notational simplicity we assume that $|f^{(m)}(x)| \leq 1$ for every $m \leq \beta+1$. This function is not necessarily a low-complexity signal, but it can be well approximated with a piecewise polynomial function. To show this, consider partitioning the $[0,1]$ interval into subintervals of size $r_n$, and approximating  the function $f$ with a polynomial of degree $\beta$ in each subinterval. Let $\hat{f}_{\beta}(x)$ denote the resulting piecewise polynomial function. It is easy to prove that $\|f-\hat{f}_{\beta}\|_{\infty} \leq \gamma r_n^{\beta+1}$. Hence, if $x$ and $x_o$ denote vectors consisting of the equispaced samples of the original signal and its piecewise polynomial approximation, respectively, it follows that $\|x-x_o\|_2 \leq \gamma \sqrt{n} r_n^{\beta+1}$. 

On the other hand the complexity of the piecewise polynomial signal is essentially proportional to $\beta/{r_n}$. Setting $r_n = n^{\frac{-2}{2\beta}}$, Theorem \ref{thm:2} proves that
$d_n = O(n^{1/\beta}\log n)$ is enough for the accurate recovery of the samples of such signals. Clearly, for $\beta <1$, this bound indicates that the number of samples we need is at the same
order as the ambient dimension. However, as $\beta$ increases fewer number of samples are required. 

Similar results hold for the piecewise  smooth functions, which are very popular in image and signal processing.


\section{Related work}\label{sec:related}
Our work is inspired by \cite{DonohoKS2002} and \cite{DoKaMe06}. \cite{DonohoKS2002} considers the well
studied problem of estimation, where the goal is to recover a vector $\theta$ from its noisy observations $s = \theta + z$, where
$z$ represents the noise in the system. It then suggests using the \textit{minimum Kolmogorov complexity
estimation} (MKCE) approach and proves that if $\theta_i \overset{iid}{\sim} \pi$, under several scenarios for the signal
and noise, the average marginal distribution of the estimate of MKCE tends to the actual posterior distribution.
On the other hand, \cite{DoKaMe06} considers the problem of compressed sensing over binary sequences.
Consider the set of all the binary sequences with Kolmogorov complexity less than or equal to $k_0$, \ie
\[
\Sc(k_0)\triangleq\{\xv: K(\xv)\leq k_0\}.
\]
Let $A$ denote a $d\times n$ binary matrix, $\xv_o=(x_1,x_2,\ldots,x_n)^T$, $\yv_o=A\xv_o$.  Consider the following  algorithm for reconstructing signal $\xv_o$ from its linear measurements $\yv_o$:
\begin{align}
{\bf \xh}(\yv_o,A) &\triangleq \argmin_{\yv_o=A\xv} K(\xv).
\end{align} 
\cite{DoKaMe06} considers this scheme and proves that $2k$ random linear binary measurements are sufficient for recovering the binary sequences in $\Sc(k_0)$ with, high probability. This result does not provide any information on the successful recovery of real signals and it does not consider the non-idealities in the signals either. Our paper settles both questions. 

As mentioned in Section \ref{sec:intro} the problem we discuss in this paper is a central problem in the field of compressed sensing \cite{Donoho1, CaRoTa06}. Several papers have considered different generalization of sparsity \cite{RichModelbasedCS, ChRePaWi10, VeMaBl02, ReFaPa10}. As mentioned
before, all these models can be considered as subclasses of the general model we consider here. However, it is worth noting that even though the recovery approach proposed in our paper is universal, since Kolmogorov complexity is not computable, it is not useful for practical purposes. 

In this paper, we considered  deterministic models for the signals. Similar extensions have been considered in the random settings  as well. For instance, \cite{WuVe10} considers the problem of recovering a memoryless process from a linear set of measurements and proves the connection between the number of measurements required and the Renyi entropy. Also, our work is in the same spirit with the minimum entropy  decoder proposed by Csiszar in \cite{Csiszar82}. He suggests a universal minimum entropy decoder, for reconstructing an iid signal from its linear measurements at a rate determined by the entropy of the source.


\section{Proof of Theorem \ref{thm:1}} \label{sec:proof1}

The following Lemma will be used in the proof of the main theorem.

\begin{lemma}[Chi-square concentration]\label{lemma:chi}
Fix $\tau>0$ and $x\in\mathds{R}^n$.  Assume that $\|x\|_2^2=1$ . Let $Z_i\triangleq\sum_{j=1}^n A_{ij}x_j$, $i=1,2,\ldots,d$. We then have,
\begin{equation}\label{eq:chisq}
\P\left( \sum_{i=1}^d  Z_i^2 -1< - \tau \right)  \leq {\rm e} ^{\frac{d}{2}(\tau + \log(1- \tau))}.
\end{equation}
\end{lemma}
\begin{proof}
Note that $\{Z_i\}_{i=1}^d$ are iid $\mathcal{N}(0,1/d)$. By Markov inequality, for any $\lambda>0$, we have
\begin{align}
\P\left( \sum_{i=1}^d  Z_i^2-1 < - \tau \right) &= \P\left(- \sum_i Z_i^2+1 > \tau \right)\nonumber\\
& \leq {\rm e}^{-\lambda \tau} \E\left[ \rm{e}^{\lambda(1- \sum Z_i^2)} \right] \nonumber \\
 &= {\rm e}^{-\lambda \tau + \lambda} \left( \E [{\rm e}^{-\lambda Z_1^2}] \right)^d \nonumber\\
 &= {\rm e}^{-\lambda \tau + \lambda} \left(1+ \frac{2\lambda}{d}\right)^{-d/2}.\label{eq:chisquarupperbound}
\end{align}
We optimize over $\lambda$ to obtain
\begin{equation}
\lambda^* = \frac{d \tau}{2(1- \tau)}. \label{eq:optlambda}
\end{equation}
If we plug \eqref{eq:optlambda} into \eqref{eq:chisquarupperbound} we obtain \eqref{eq:chisq}.
\end{proof}

\begin{proof}[Proof of Theorem \ref{thm:1}] 
Let $e_m^n=x_o^n-[x_o^n]_m$ and $\eh_m^n=\xh_o^n-[\xh_o^n]_m$ denote the quantization errors of the original and the reconstructed signals, respectively. Since both $Ax_o^n=y_o$ and $A\xh^n=y_o$, it follows that 
\begin{align}
A([x_o^n]_m+e_m^n)&=A([\xh_o^n]_m+\eh_m^n)\nonumber
\end{align}
and 
\begin{align}
A([x_o^n]_m-[\xh_o^n]_m)&=A(\eh_m^n-e_m^n).
\end{align}
On the other hand, since $|y-[y]_m|\leq 2^{-m}$, for each $y\in[0,1]$,  we have
\[
\|\eh_m^n-e_m^n\|_2^2\leq n2^{-2m+1}.
\] Hence,
\begin{align}
\|A([x_o^n]_m-[\xh_o^n]_m)\|_2&=\|A(\eh_m^n-e_m^n)\|_2 \nonumber\\
&\leq \sigma_{\max}(A) \sqrt{n2^{-2m+1}}.\label{eq:upper_bd}
\end{align}

Since, by assumption, \eqref{eq:limited_kol} holds for $x_o$, for each $\d>0$, there exists $N_{\d}$, such that for any $n>N_{\d}$,
\begin{align}
{K^{[\cdot]_{m}}(x_{o}^n) \over m} \leq \kappa+\d\label{eq:kol_xo}
\end{align}
Since $\xh_o^n$ is the solution of \eqref{eq:alg},
\begin{align}
 K^{[\cdot]_{m}}(\xh_{o}^n) \leq K^{[\cdot]_{m}}(x_{o}^n).\label{eq:kol_xho}
\end{align}
Moreover, 
\begin{align}
 K([x_{o}^n]_m-[\xh_o^n]_m) \leq K^{[\cdot]_{m}}(x_{o}^n)+ K^{[\cdot]_{m}}(\xh_{o}^n)+C, \label{eq:kol_dif}
\end{align}
where $C$ is a constant independent of all the other variables in the problem \cite{cover}. Combining \eqref{eq:kol_xo}, \eqref{eq:kol_xho} and \eqref{eq:kol_dif} yields
\begin{align}
 K([x_{o}^n]_m-[\xh_o^n]_m) \leq 2(\kappa+\d)m+C.
\end{align}
If for each sequence $y^n$ with $K^{[\cdot]_m}(y^n) \leq  2(\kappa+\d)m+C$, $\|A [y^n]_m\|_2 \geq \tau \|[y^n]_m\|_2$, for some fixed $\tau>0$, then from \eqref{eq:upper_bd}
\begin{align}
\|x_o^n - \xh_o^n\|_2 &= \left\|[x_o^n]_m+e_m^n - [\xh_o^n]_m-\eh_m^n\right\|_2 \nonumber\\
& \leq\left\|[x_o^n]_m - [\xh_o^n]_m\|_2  + \|e^n_m - \eh^n_m\right\|_2 \nonumber\\
&\leq \tau^{-1}\sigma_{\max}(A) \sqrt{n2^{-2m+1}}+\sqrt{n2^{-2m+1}}\nonumber\\
&\leq (\tau^{-1}\sigma_{\max}(A) +1)\sqrt{n2^{-2m+1}}.
\end{align}

Define the events $\Ec_1^{(n)}$ and  $\Ec_2^{(n)}$ as
\begin{align}
&\Ec_1^{(n)}\triangleq\{A_{d\times n}:\nonumber\\
& \nexists \; y^n; K^{[\cdot]_m}(y^n)\leq 2(\kappa+\d)+C,  \|A y^n\|_2 < \tau \|y^n\|_2 \},\label{eq:E1}
\end{align}
and
\begin{align}
\Ec_2^{(n)}\triangleq \left\{A_{d\times n}: \sigma_{max}(A) - 1 - \sqrt{\frac{n}{d}} < t\right\},\label{eq:E2}
\end{align}
for some $t>0$.

Using these definitions plus the union bound, it follows that 
\begin{align}
\P\left(\| x_o^n-\xh_o^n\|_2 >\e\right)=&\P\left(\| x_o^n-\xh_o^n\|_2 >\e,\Ec_1^{(n)} \cap\Ec_2^{(n)}\right)\nonumber\\
&+\P\left(\| x_o^n-\xh_o^n\|_2 >\e,(\Ec_1^{(n)}\cap \Ec_2^{(n)})^c\right)\nonumber\\
\leq &\P\left(\| x_o^n-\xh_o^n\|_2 >\e,\Ec_1^{(n)} \cap\Ec_2^{(n)}\right)\nonumber\\
&+\P\left((\Ec_1^{(n)}\cap \Ec_2^{(n)})^c\right)\nonumber\\
\leq &\P\left(\| x_o^n-\xh_o^n\|_2 >\e,\Ec_1^{(n)} \cap\Ec_2^{(n)}\right)\nonumber\\
&+\P\left(\Ec_1^{(n),c}\right)+\P\left(\Ec_2^{(n),c}\right).
\end{align}

If $A\in \Ec_1^{(n)}\cap \Ec_2^{(n)}$, then from \eqref{eq:upper_bd}
\begin{align}
\|x_o^n - \xh_o^n\|_2 \leq \left(\tau^{-1}(\sqrt{n\over d}+1+t) +1\right)\sqrt{n2^{-2m+1}}.
\end{align}
Since, by assumption, $m=m_n=\lceil\log n\rceil$ and $d=d_n=\lceil\kappa \log n\rceil$, if $n$ large enough,
\begin{align}
\left(\tau^{-1}(\sqrt{n\over d}+1+t) +1\right)\sqrt{n2^{-2m+1}} <\e.
\end{align}
Hence, for $n$ large enough
\begin{align}
\P\left(\| x_o^n-\xh_o^n\|_2 >\e,\Ec_1^{(n)} \cap\Ec_2^{(n)}\right)=0.
\end{align}

On the other hand, by Lemma \ref{lemma:chi}, for each sequence $x^n\in\mathds{R}^n$,
\begin{align}
\P\{\|Ax^n\|_2^2 \leq \tau \|x^n\|_2^2 \} &= \P \{\|A\frac{x^n}{\|x^n\|_2}\|_2^2 \leq \tau^2 \}\nonumber\\
& \leq {\rm e}^{\frac{d}{2}(1- \tau^2+ 2 \log \tau) }.
\end{align}
Therefore,
\begin{align}
&\P\left(\Ec_1^{(n),c}\right)=\nonumber\\
& \P\left\{\exists \; y^n: K^{[\cdot]_m}(y^n)\leq 2(\kappa+\d)m+C,  \|A y^n\|^2_2 < \tau \|y^n\|^2_2 \right\}\nonumber\\
& \leq 2^{2(\kappa+\d)m+C} \rm{e}^{-\frac{d}{2}(1- \tau^2+ 2 \log \tau) }.
\end{align}
If we set $\tau =0.04$ and $d = \lceil \kappa \log n \rceil$ it is simple to see that this probability goes to zero. 
Finally, we can use the concentration of Lipschitz function of a Gaussian random vector to prove \cite{CaTa05}
\begin{align}
\P\left(\Ec_2^{(n),c}\right) &= \P\left(\sigma_{max}(A) - 1 - \sqrt{\frac{n}{d}} > t\right)\nonumber\\ 
&\leq {\rm e}^{-d t^2/2}.
\end{align}
Setting $t$ to a constant and $d = \lceil \kappa \log n \rceil$ proves that this probability also goes to zero. 
\end{proof}


\section{Proof of Theorem \ref{thm:2}} \label{sec:proof2}

Let $x_o^n=[x_o^n]_m+e_m^n$, $\xt_o^n=[\xt_o^n]_m+\et_m^n$, and $\xh_o^n=[\xh_o^n]_m+\eh_m^n$.

Note that  since $ \|A\xt_o^n -y_o^n\|_2 =\| A(\xt_o^n-x_o^n) \|_2 \leq  \sigma_{max}(A) \epsilon_n$, $\xt_o^n$ is also a feasible solution.  Therefore, since $\xt_o^n$ and $\xh_o^n$ are both feasible, by triangle inequality, 
\begin{align}
\|A\xt_o^n - A\xh_o^n \|_2&=\|A\xt_o^n-y_o^n - (A\xh_o^n-y_o^n) \|_2 \nonumber\\
&\leq 2\sigma_{\max}(A)\e_n.\label{eq:normeq1}
\end{align}

Again, by triangle inequality, 
\begin{align}
&\|A\xt_o^n - A\xh_o^n \|_2 \nonumber\\
&=\|A([\xt_o^n]_m+\et_m^n)- A([\xh_o^n]_m+\eh_m^n) \|_2 \nonumber\\
& \geq \|A([\xt_o^n]_m-[\xh_o^n]_m)\|_2 -  \|A([\et^n]_m-[\eh^n]_m)\|_2\nonumber\\
& \geq \|A([\xt_o^n]_m-[\xh_o^n]_m)\|_2 - \sigma_{max}(A) \|[\et^n]_m-[\eh^n]_m\|_2\nonumber\\
& \geq \|A([\xt_o^n]_m-[\xh_o^n]_m)\|_2 - \sigma_{max}(A) \sqrt{n 2^{-2m+1}}.\label{eq:normeq2}
\end{align}
Combining \eqref{eq:normeq1} and \eqref{eq:normeq2}, it follows
\begin{align}
\|A([\xt_o^n]_m - [\xh_o^n]_m) \|_2 \leq \sigma_{max}(A) \sqrt{n 2^{-2m+1}}+2\sigma_{\max}(A)\e_n.
\end{align}

Since both $\xt_o^n$ and $\xh_o^n$ are feasible, and $\xh_o^n$ is the optimizer of \eqref{eq:alg_model_mismatch}, we have
\begin{align}
K^{[\cdot]_m}(\xh_o^n) &\leq K^{[\cdot]_m}(\xt_o^n) \leq m (\kappa_n+\d),
\end{align}
and therefore
\begin{align}
K^{[\cdot]_m}(\xh_o^n-\xt_o^n) \leq m 2(\kappa_n+\d)+C,
\end{align}
where $C$ is a constant independent of $m$ and $n$.

Consider defining the events $\Ec_1$ and $\Ec_2$ as done in \eqref{eq:E1} and \eqref{eq:E2}, in the proof of Theorem \ref{thm:1}. Then, using the same argument used in that proof,
\begin{align}
\P\left(\| x_o^n-\xh_o^n\|_2 >\e\right)
\leq &\P\left(\| x_o^n-\xh_o^n\|_2 >\e,\Ec_1^{(n)} \cap\Ec_2^{(n)}\right)\nonumber\\
&+\P\left(\Ec_1^{(n),c}\right)+\P\left(\Ec_2^{(n),c}\right).
\end{align}

However, our choice of parameters guarantees that for large enough $n$, $\P(\| x_o^n-\xh_o^n\|_2 >\e,\Ec_1^{(n)} \cap\Ec_2^{(n)})=0$, and moreover, $\P(\Ec_1^{(n),c})$ and $\P(\Ec_1^{(n),c})$ both go to 0 as $n$ grows to infinity. 

%


\section{Conclusion}

In this paper, we consider the problem of recovering structured signals from their linear measurements.  We use the Komogorov complexity of the quantized signal as a universal measure of complexity that covers many different examples explored in compressed sensing literature and related areas. We then show that, if we consider low-complexity signals, the minimum complexity pursuit scheme inspired by the Occam's razor recovers the simplest solution of a set of random linear measurements.
In fact, we prove that the number of measurements required is proportional to the complexity and logarithmically to the ambient dimension of the signal. We  also consider more practical scenarios where the signal is not `simple' but is `close' to a low complexity signal. We show that even in such cases following minimum complexity pursuit algorithm
provides a good estimate of the signal from much fewer samples than the ambient dimension of the signal.

As mentioned in the paper, Kolmogorov complexity of a sequence is not computable. However, currently we are working on deriving implementable schemes by replacing Kolmogorov complexity by computable measures such as  miminimum description length \cite{Rissanen86}.

\bibliographystyle{unsrt}
\bibliography{../myrefs}

\end{document}